\documentclass[conference, 10pt]{IEEEtran}
\usepackage[dvips]{graphicx}
\usepackage{amssymb,amsmath}
\IEEEoverridecommandlockouts
\newtheorem{theorem}{Theorem}[section]
\newtheorem{lemma}[theorem]{Lemma}

\newtheorem{definition}{Definition}
\newcommand{\ie}{\emph{i.e. }}
\newcommand{\etal}{\textit{et al.}}
\begin{document}

\title{An Algebraic Watchdog for Wireless Network Coding}
\author{\authorblockN{MinJi Kim\authorrefmark{1}, Muriel M\'{e}dard\authorrefmark{1}, Jo\~{a}o Barros\authorrefmark{2}, and Ralf K\"{o}tter\authorrefmark{3}}
\authorblockA{\\ \authorrefmark{1}Research Laboratory of Electronics\\
Massachusetts Institute of Technology\\
Cambridge, MA 02139, USA\\
Email: \{minjikim, medard\}@mit.edu}\vspace*{0.2cm}
\authorblockA{\authorrefmark{2}Instituto de Telecommunica\c{c}\~{o}es\\
Departamento de Engenharia Electrot\'{e}cnica e de Computadores\\
Faculdade de Engenharia da Universidade do Porto, Portugal\\
Email: jbarros@fe.up.pt}\vspace*{0.2cm}
\authorblockA{\authorrefmark{3}Institute for Communications Engineering\\
Technische Universit\"{a}t M\"{u}nchen, Munich, Germany}\vspace*{-0.5cm}
\thanks{This material is based upon work under a subcontract \#069145 issued by BAE Systems National Security Solutions, Inc. and supported by the DARPA and the Space and Naval Warfare System Center, San Diego under Contract No. N66001-08-C-2013.}
\thanks{$^\ddag$Ralf K\"{o}tter passed away earlier this year.}
}

\maketitle

\begin{abstract}
In this paper, we propose a scheme, called the \emph{algebraic watchdog} for wireless network coding, in which nodes can detect malicious behaviors probabilistically, police their downstream neighbors locally using overheard messages, and, thus, provide a secure global \emph{self-checking network}. Unlike traditional Byzantine detection protocols which are \emph{receiver-based}, this protocol gives the senders an active role in checking the node downstream. This work is inspired by Marti \etal's \emph{watchdog-pathrater}, which attempts to detect and mitigate the effects of routing misbehavior.

As the first building block of a such system, we focus on a two-hop network. We present a graphical model to understand the inference process nodes execute to police their downstream neighbors; as well as to compute, analyze, and approximate the probabilities of misdetection and false detection. In addition, we present an algebraic analysis of the performance using an hypothesis testing framework, that provides exact formulae for probabilities of false detection and misdetection.
\end{abstract}
\IEEEpeerreviewmaketitle

\section{Introduction}
%
%


There have been numerous contributions to secure wireless networks, including key management, secure routing, Byzantine detection, and various protocol designs (for a general survey on this topic, see \cite{hubaux}). We focus on Byzantine detection. The traditional approach is \emph{receiver-based} -- \ie the receiver of the corrupted data detects the presence of an upstream adversary. However, this detection may come too late as the adversary is partially successful in disrupting the network (even if it is detected). It has wasted network bandwidth, while the source is still unaware of the need for retransmission.

Reference \cite{marti} introduces a protocol for routing wireless networks, called the \emph{watchdog and pathrater}, in which upstream nodes police their downstream neighbors using \emph{promiscuous monitoring}. Promiscuous monitoring means that if a node $A$ is within range of a node $B$, it can overhear communication to and from $B$ even if those communication do not directly involve $A$. This scheme successfully detects adversaries and removes misbehaving nodes from the network by dynamically adjusting the routing paths. However, the protocol requires a significant overhead (12\% to 24\%) owing to increased control traffic and numerous cryptographic messages.

Our goal is to design/analyze a watchdog-inspired protocol for wireless networks using network coding. Network coding \cite{ahlswede}\cite{algebraic} is advantageous as it not only increases throughput and robustness against failures and erasures but also it is resilient in dynamic/unstable networks where state information may change rapidly or may be hard to obtain. Taking advantage of the wireless setting, we propose a scheme for coded networks, in which nodes can verify probabilistically, and police their neighbors locally using promiscuous monitoring. Our ultimate goal is a robust \emph{self-checking network}. In this paper, we present the first building block of a such system, and analyze the algebraic watchdog protocol for a two-hop network.

The paper is organized as follows. In Section \ref{sec:background}, we present the background and related material. In Section \ref{sec:problemstatement}, we introduce our problem statement and network model. In Section \ref{sec:example}, we analyze the protocol for a simple two-hop network, first algebraically in Section \ref{sec:algebraicapproach} and then graphically in Section \ref{sec:graphicalmodel}. In Section \ref{sec:conclusion}, we summarize our contribution and discuss some future work.

\section{Background and Definitions}\label{sec:background}
\subsection{Secure Network Coding}\label{sec:securenc}
Network coding, first introduced in \cite{ahlswede}, allows algebraic mixing of information in the intermediate nodes. This mixing has been shown to have numerous performance benefits. It is known that network coding maximizes throughput \cite{ahlswede}, as well as robustness against failures \cite{algebraic} and erasures \cite{reliable}. However, a major concern for network coding system is its vulnerability to Byzantine adversaries. A single corrupted packet generated by a Byzantine adversary can contaminate all the information to a destination, and propagate to other destinations quickly. For example, in random linear network coding \cite{reliable}, one corrupted packet in a generation (\ie a fixed set of packets) can prevent a receiver from decoding any data from that generation even if all the other packets it has received are valid.

There are several papers that attempt to address this problem. One approach is to correct the errors injected by the Byzantine adversaries using \emph{network error correction} \cite{errorcorrection}. They bound the maximum achievable rate in an adversarial setting, and generalizes the Hamming, Gilbert-Varshamov, and Singleton bounds. Jaggi \etal \cite{resilient} propose a distributed, rate-optimal, network coding scheme for multicast network that is resilient in the presence of Byzantine adversaries for sufficiently large field and packet size. Reference \cite{subspace} generalizes \cite{resilient} to provide correction guarantees against adversarial errors for any given field and packet size. In \cite{milcom}, Kim \etal compare the cost and benefit associated with these Byzantine detection schemes in terms of transmitted bits by allowing nodes to employ the detection schemes to drop polluted data.


\subsection{Secure Routing Protocol: Watchdog and Pathrater}

The problem of securing networks in the presence of Byzantine adversaries has been studied extensively, e.g. \cite{perlman},\cite{liskov},\cite{Lamport}. The \emph{watchdog and pathrater} \cite{marti} are two extensions to the Dynamic Source Routing \cite{dsr} protocol that attempt to detect and mitigate the effects of routing misbehavior. The watchdog detects misbehavior based on promiscuous monitoring of the transmissions of the downstream node to confirm if this relay correctly forwards the packets it receives. If a node bound to forward a packet fails to do so after a certain period of time, the watchdog increments a failure rating for that node and a node is deemed to be misbehaving when this failure rating exceeds a certain threshold. The pathrater then uses the gathered information to determine the best possible routes by avoiding misbehaving nodes. This mechanism, which does not punish these nodes (it actually relieves them from forwarding operations), provides an increase in the throughput of networks with misbehaving nodes.

\subsection{Hypothesis Testing}\label{sec:hypothesis}

Hypothesis testing is a method of deciding which of the two hypotheses, denoted $H_0$ and $H_1$, is true, given an observation denoted as $U$. In this paper, $H_0$ is the hypothesis that $R$ is well-behaving, $H_1$ is that $R$ is malicious, and $U$ is the information gathered from overhearing. The observation $U$ is distributed differently depending whether $H_0$ or $H_1$ is true, and these distributions are denoted as $P_{U|H_0}$ and $P_{U|H_1}$ respectively.

An algorithm is used to choose between the hypotheses given the observation $U$. There are two types of error associated with the decision process:
\begin{itemize}
\item{\it{Type 1 error, False detection}}: Accepting $H_1$ when $H_0$ is true (\ie considering a well-behaving $R$ to be malicious), and the probability of this event is denoted $\gamma$.
\item{\it{Type 2 error, Misdetection}}: Accepting $H_0$ when $H_1$ is true (\ie considering a malicious $R$ to be well-behaving), and the probability of this event is denoted $\beta$.
\end{itemize}
The Neyman-Pearson theorem gives the optimal decision rule that given the maximal tolerable $\beta$, we can minimize $\gamma$ by accepting hypothesis $H_0$ if and only if $\log \frac{P_{U|H_0}}{P_{U|H_1}} \geq t$ for some threshold $t$ dependant on $\gamma$. For more thorough survey on hypothesis testing in the context of authentication, see \cite{hypothesis}.

\subsection{Notations and definitions}

We shall use elements from a field, and their bit-representation. To avoid confusion, we use the same character in italic font (\ie $x$) for the field element, and in bold font (\ie $\mathbf{x}$) for the bit-representation. We use underscore bold font (\ie $\mathbf{\underline{x}}$) for vectors. For arithmetic operations in the field, we shall use the conventional notation (\ie $+, -, \cdot
$). For bit-operation, we shall use $\oplus$ for addition, and $\otimes$ for multiplication.

We also require polynomial hash functions defined as follows (for a more detailed discussion on this topic, see \cite{hash}).\vspace*{-0.45cm}
\begin{definition}[\textit{\textbf{Polynomial hash functions}}] For a finite field $\mathbf{F}$ and $d \geq 1$, the class of polynomial hash functions on $\mathbf{F}$ is defined as follows:
\[
\mathcal{H}^d(\mathbf{F}) = \{h_a | a = \langle a_0, ..., a_d \rangle \in \mathbf{F}^{d+1}\},
\]
where $h_a(x) = \sum_{i=0}^d a_i x^i$ for $x\in \mathbf{F}$.
\end{definition}

\section{Problem Statement}\label{sec:problemstatement}
We model a wireless network with a hypergraph $G = (V, E_1, E_2)$, where $V$ is the set of the nodes in the network, $E_1$ is the set of hyperedges representing the connectivity (wireless links), and $E_2$ is the set of hyperedges representing the interference. We use the hypergraph to capture the broadcast nature of the wireless medium. If $(v_1, v_2) \in E_1$ and $(v_1,v_3)\in E_2$ where $v_1, v_2, v_3\in V$, then there is an intended transmission from $v_1$ to $v_2$, and $v_3$ can overhear this transmission (possibly incorrectly). There is a certain transition probability associated with the interference channels known to the nodes, and we model them with binary channels.

A node $v_i\in V$ transmits coded information $x_i$ by transmitting a packet $\mathbf{\underline{p_i}}$, where $\mathbf{\underline{p_i}}=[\mathbf{a_i}, \mathbf{h_{I_i}}, \mathbf{h_{x_i}}, \mathbf{x_i}]$ is a $\{0,1\}$-vector. A valid packet $\mathbf{\underline{p_i}}$ is defined as below:
\begin{itemize}
\item $\mathbf{a_i}$ corresponds to the coding coefficients $\alpha_j$, $j \in I_i$, where $I_i \subseteq V$ is the set of nodes adjacent to $v_i$ in $E_1$,
\item $\mathbf{h_{I_i}}$ corresponds to the hash $h(x_j)$, $v_j \in I_i$ where $h(\cdot)$ is a $h$-bit polynomial hash function,
\item $\mathbf{h_{x_i}}$ corresponds to the polynomial hash $h(x_i)$,
\item $\mathbf{x_i}$ is the $n$-bit representation of $x_i = \sum_{j \in I} \alpha_j x_j$.
\end{itemize}

We assume that the hash function used, $h(\cdot)$, is known to all nodes, including the adversary. In addition, we assume that $\mathbf{a_i}$, $\mathbf{h_{I_i}}$ and $\mathbf{h_{x_i}}$ are part of the header information, and are sufficiently coded to allow the nodes to correctly receive them even under noisy channel conditions. Therefore, if a node overhears the transmission of $\mathbf{\underline{p_i}}$, it may not be able to correctly receive $x_i$, but it receives $\alpha_j$ and $h(x_j)$ for $v_j \in I_i$, and $h(x_i)$. Protecting the header sufficiently will of course induce some overhead, but the assumption remains a reasonable one to make. First, the header is smaller than the message itself. Second, even in the routing case, the header and the state information must to be coded sufficiently. Third, the hashes $\mathbf{h_{I_i}}$ and $\mathbf{h_{x_i}}$ are contained within one hop -- \ie a node that receives $\mathbf{\underline{p_i}}=[\mathbf{a_i}, \mathbf{h_{I_i}}, \mathbf{h_{x_i}}, \mathbf{x_i}]$ does not need to repeat $\mathbf{h_{I_i}}$, thus sending only $\mathbf{h_{x_i}}$. Therefore, the overhead associated with the hashes is proportional to the in-degree of a node, and does not accumulate with the routing path length.

\begin{figure}[h!]
\begin{center}
\includegraphics[width=0.26\textwidth]{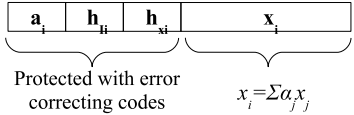}
\end{center}
\vspace*{-0.5cm}
\caption{A valid packet $\mathbf{\underline{p_i}}$ sent by well-behaving $R$}\label{fig:packet}
\end{figure}

Assume that $v_i$ transmits $\mathbf{\underline{p_i}}=[\mathbf{a_i}, \mathbf{h_{I_i}}, \mathbf{h_{x_i}}, \mathbf{\hat{x}_i}]$, where $\mathbf{\hat{x}_i} = \mathbf{x_i} \oplus \mathbf{e}$,  $\mathbf{e} \in \{0,1\}^n$. If $v_i$ is misbehaving, then $\mathbf{e} \ne 0$. It is important to note that the adversary can choose any $\mathbf{e}$; thus, the adversary can choose the message $\mathbf{\hat{x}_i}$. Our goal is to detect with high probability when $\mathbf{e} \ne 0$. Note that even if $|\mathbf{e}|$ is small (\ie the hamming distance between $\mathbf{\hat{x}_i}$ and $\mathbf{x_i}$ is small), the algebraic interpretation of $\mathbf{\hat{x}_i}$ and $\mathbf{x_i}$ may differ significantly. For example, consider $n=4$, $\mathbf{\hat{x}_i} =[0000]$, and $\mathbf{x_i} = [1000]$. Then, $\mathbf{e} = [1000]$ and $|\mathbf{e}| = 1$. However, the algebraic interpretation of $\mathbf{\hat{x}_i}$ and $\mathbf{x_i}$ are 0 and 8. Thus, even a single bit flip can alter the message very significantly.

%
%
%
%

Our goal is to explore an approach to detect and prevent malicious behaviors in wireless networks using network coding. The scheme takes advantage of the wireless setting, where neighbors can overhear others' transmissions albeit with some noise, to verify probabilistically that the next node in the path is behaving given the overheard transmissions.

\section{Two-hop network}\label{sec:example}

Consider a network (or a small neighborhood of nodes in a larger network) with nodes $v_1, v_2, ... v_m,$ $v_{m+1}$, $v_{m+2}$. Nodes $v_i$, $i\in [1,m]$, want to transmit $x_i$ to $v_{m+2}$ via $v_{m+1}$. A single node $v_i$, $i\in [1,m]$, cannot check whether $v_{m+1}$ is misbehaving or not even if $v_i$ overhears $\mathbf{x_{m+1}}$, since without any information about $x_j$ for $j \in [1,m]$, ${x}_{m+1}$ is completely random to $v_i$. On the other hand, if $v_i$ knows $x_{m+1}$ and $x_j$ for all $j \in [1, m]$, then $v_i$ can verify that $v_{m+1}$ is behaving with certainty; however, this requires at least $m-1$ additional reliable transmissions to $v_i$.

\begin{figure}[h!]
\begin{center}
\includegraphics[width=0.38\textwidth]{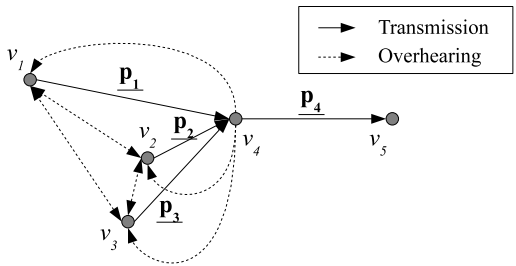}
\end{center}
\vspace*{-0.5cm}
\caption{A wireless network with $m=3$.}\label{fig:general}
\end{figure}

Therefore, we take advantage of the wireless setting, in which nodes can overhear their neighbors' transmissions. In Figure \ref{fig:general},  we use the solid lines to represent the intended channels $E_1$, and dotted lines for the interference channels $E_2$ which we model with binary channels as mentioned in Section \ref{sec:problemstatement}. Each node checks whether its neighbors are transmitting values that are consistent with the gathered information. If a node detects that its neighbor is misbehaving, then it can alert other nodes in the network and isolate the misbehaving node.

\begin{figure}[tbp]
\begin{center}
\includegraphics[width=0.39\textwidth]{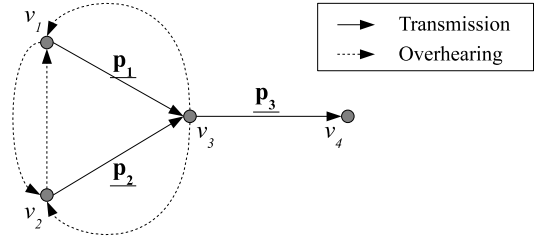}
\end{center}
\vspace*{-0.5cm}
\caption{A wireless network with $m=2$.}\label{fig:2sources}
\end{figure}
%


As outlined in Section \ref{sec:hypothesis}, we denote the hypothesis that $R$ is well-behaving by $H_0$, and $H_1$ corresponds to that of a malicious $R$. In the next subsections, we shall use an example with $m=2$, as shown Figure \ref{fig:2sources}, to introduce the graphical model which explains how a node $v_i$ checks its neighbor's behavior. Then, we use an algebraic approach to analyze/compute $\gamma$ and $\beta$ for this example network.

%

\subsection{Graphical model approach}\label{sec:graphicalmodel}

In this section, we present a graphical approach to model this problem systematically, and to explain how a node may check its neighbors. This approach may be advantageous as it lends easily to already existing graphical model algorithms as well as some approximation algorithms.

We shall consider the problem from $v_1$'s perspective. As shown in Figure \ref{fig:trellis}, the graphical model has four layers: Layer 1 contains $2^{n+h}$ vertices, each representing a bit-representation of $[\mathbf{\tilde{x}_2, h(x_2)}]$; Layer 2 contains $2^n$ vertices, each representing a bit-representation of $\mathbf{x_2}$; Layer 3 contains $2^n$ vertices corresponding to $\mathbf{x_3}$; and Layer 4 contains $2^{n+h}$ vertices corresponding to  $[\mathbf{\tilde{x}_3, h(x_3)}]$. Edges exist between adjacent layers as follows:
\begin{itemize}
\item{\it{Layer 1 to Layer 2:}} An edge exists between a vertex $[\mathbf{v,u}]$ in Layer 1 and a vertex $\mathbf{w}$ in Layer 2 if and only if $\mathbf{h(w) = u}$. The edge weight is normalized such that the total weight of edges leaving $[\mathbf{v, u}]$ is 1, and the weight is proportional to:
    \vspace*{-.08cm}
     \[\mathbf{P}(\mathbf{v}| \text{ Channel statistics and } \mathbf{w} \text{ is the original message}),\] \vspace*{-.08cm}which is the probability that the inference channel outputs message $\mathbf{v}$ given an input message $\mathbf{w}$.
\item{\it{Layer 2 to Layer 3:}} The edges represent a permutation. A vertex $\mathbf{v}$ in Layer 2 is adjacent to a vertex $\mathbf{w}$ in Layer 3 if and only if $w = c+\alpha_2 v$, where $c = \alpha_1 x_1$ is a constant, $\mathbf{v}$ and $\mathbf{w}$ are the bit-representation of $v$ and $w$, respectively. The edge weights are all 1.
\item{\it{Layer 3 to Layer 4:}}  An edge exists between a vertex $\mathbf{v}$ in Layer 3 and a vertex $[\mathbf{w,u}]$ in Layer 4 if and only if $\mathbf{h(v) = u}$. The edge weight is normalized such that the total weight leaving $\mathbf{v}$ is 1, and is proportional to:
        \vspace*{-.08cm}
     \[\mathbf{P}(\mathbf{w}| \text{ Channel statistics and } \mathbf{v} \text{ is the original message}).\]
\end{itemize}

\begin{figure}[tbp]
\begin{center}
\includegraphics[width=0.48\textwidth]{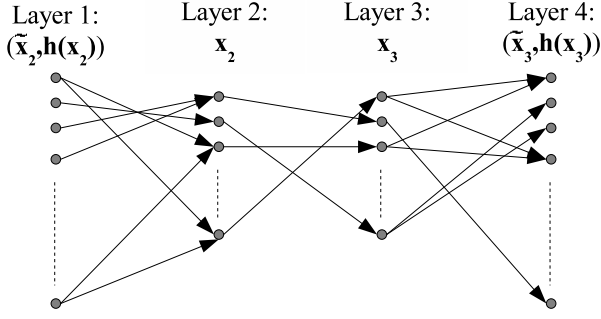}
\end{center}
\vspace*{-0.5cm}
\caption{A graphical model from $v_1$'s perspective}\label{fig:trellis}
\end{figure}

Node $v_1$ overhears the transmissions from $v_2$ to $v_3$ and from $v_3$ to $v_4$; therefore, it receives $[\mathbf{\tilde{x}_2, h(x_2)}]$ and $[\mathbf{\tilde{x}_3, h(x_3)}]$, corresponding to the \emph{starting point} in Layer 1 and the \emph{destination point} in Layer 4 respectively. By computing the sum of the product of the weights of all possible paths between the starting and the destination points, $v_1$ computes the probability that $v_3$ is consistent with the information gathered.

This graphical model illustrates sequentially and visually the inference process $v_1$ executes. In addition, the graphical approach may be extend to larger networks. Cascading multiple copies of the graphical model may allow us to systematically model larger networks with multiple hops as well as $m\geq 3$. (Note that when $m$ increases, the graphical model changes into a family of graphs; while when $n$ increases, the size of each Layer increases.) Furthermore, by using approximation algorithms and pruning algorithms, we may be able to simplify the computation as well as the structure of the graph.

%
%
%

\subsection{Algebraic approach}\label{sec:algebraicapproach}

Consider $v_1$. By assumption, $v_1$ correctly receives $\mathbf{a_2}$,  $\mathbf{a_3}$, $\mathbf{h_{I_2}}$, $\mathbf{h_{I_3}}$, $\mathbf{h_{x_2}}$, and $\mathbf{h_{x_3}}$. In addition, $v_1$ receives $\mathbf{\tilde{x}_2 = x_2 + e'}$ and $\mathbf{\tilde{x}_3 = x_3 + e''}$, where $\mathbf{e'}$ and $\mathbf{e''}$ are outcomes of the interference channels. Given $\mathbf{\tilde{x}_j}$ for $j = \{2,3\}$ and the transition probabilities, $v_1$ computes $r_{j\rightarrow 1}$ such that the sum of the probability that the interference channel from $v_j$ and $v_1$ outputs $\mathbf{\tilde{x}_j}$ given $\mathbf{x} \in B(\mathbf{\tilde{x}_j}, r_{j\rightarrow 1})$ is greater or equal to $1-\epsilon$ where $\epsilon$ is a constant, and $B(\mathbf{x}, r)$ is a $n$-dimensional ball of radius $r$ centered at $\mathbf{x}$. Now, $v_1$ computes $\tilde{X}_j = \{\mathbf{x}\ |\ h(x) = h(x_j)\} \cap B(\mathbf{\tilde{x}_j}, r_{j\rightarrow 1})$ for $j = \{2, 3\}$. Then, $v_1$ computes $\alpha_1 x_1 + \alpha_2 \hat{x}$ for all $\mathbf{\hat{x}} \in \tilde{X}_2$. Then, $v_1$ intersects $\tilde{X}_3$ and the computed $\alpha_1 x_1 + \alpha_2 \hat{x}$'s. If the intersection is empty, then $v_1$ claims that $R$ is misbehaving.

We explain the inference process described above using the graphical model introduced in Section \ref{sec:graphicalmodel}. The set $\{\mathbf{x}\ |\ h(x) = h(x_2)\}$ represents the Layer 2 vertices reachable from the starting point ($[\mathbf{\tilde{x}_2, h(x_2)}]$ in Layer 1), and $\tilde{X}_2$ is a subset of the reachable Layer 2 vertices such that the total edge weight (which corresponds to the transition probability) from the starting point is greater than $1-\epsilon$. Then, computing $\alpha_1 x_1 + \alpha_2 \hat{x}$ represents the permutation from Layers 2 to 3. Finally, the intersection with $\tilde{X}_3$ represents finding a set of Layer 3 vertices such that they are adjacent to the destination point ($[\mathbf{\tilde{x}_3, h(x_3)}]$ in Layer 4) and their total transition probability to the destination point is greater than $1-\epsilon$.




Note that a malicious $v_3$ would not inject errors in $\mathbf{h_{x_3}}$ only, because the destination $v_4$ can easily verify if $\mathbf{h_{x_3}}$ is equal to $h(\mathbf{x_3})$. Therefore, $\mathbf{h_{x_3}}$ and $\mathbf{x_3}$ are consistent. In addition, $v_3$ would not inject errors in $\mathbf{h_{x_j}}$, $j \in I_3$, as each node $v_j$ can verify the hash of its message. On the other hand, a malicious $v_3$ can inject errors in $\mathbf{a_3}$, forcing $v_4$ to receive incorrect coefficients $\tilde{\alpha}_j$'s instead of $\alpha_j$'s. However, any error introduced in $\mathbf{a_3}$ can be translated to errors in $\mathbf{x_3}$ by assuming that $\tilde{\alpha}_j$'s are the correct coding coefficients. Therefore, we are concerned only with the case in which $v_3$ introduces errors in $\mathbf{x_3}$ (and therefore, in $\mathbf{h_{x_3}}$ such that $\mathbf{h_{x_3}} = h(\mathbf{x_3})$).

\begin{lemma}
 For $n$ sufficiently large, the probability of false detection, $\gamma \leq \epsilon$ for any arbitrary small constant $\epsilon$.
\end{lemma}
\begin{proof}
Assume that $v_3$ is not malicious, and transmits $\mathbf{x_3}$ and $\mathbf{h_{x_3}}$ consistent with $v_4$'s check. Then, for $n$ sufficiently large, $v_1$ can choose $r_{2\rightarrow 1}$ and $r_{3\rightarrow 1}$ such that the probability that the bit representation of $x_3 = \alpha_1 x_1 + \alpha_2 x_2$ is in $\tilde{X}_3$ and the probability that $\mathbf{x_2} \in \tilde{X}_2$ are greater than $1-\epsilon$. Therefore, $\tilde{X}_3 \cap \{\alpha_1 x_1 + \alpha_2 \hat{x}\ |\ \forall \mathbf{\hat{x}} \in \tilde{X}_2\} \ne \emptyset$ with probability arbitrary close to 1. Therefore, a well-behaving $v_3$ passes $v_1$'s check with probability at least $1-\epsilon$. Thus, $\gamma \leq \epsilon$.
\end{proof}

\begin{lemma}\label{thm:s1}
$\mathbf{P}($A malicious $v_3$ is undetected from $v_1$'s perspective$)$ is:
\[
\min\biggl\{1, \frac{\sum_{k=0}^{r_{1\rightarrow 2}} \binom{n}{k}}{2^{(h+n)}}\cdot \frac{\sum_{k=0}^{r_{2\rightarrow 1}} \binom{n}{k}}{2^{(h+n)}}\cdot \frac{\sum_{k=0}^{r_{3\rightarrow 1}} \binom{n}{k}}{2^{h}}\biggl\}.
\]
\end{lemma}
\begin{proof}
Assume that $v_3$ is malicious and injects errors into $\mathbf{x_3}$. Consider an element $\mathbf{z} \in \tilde{X}_3$, where $z = \alpha_1 x_1 + \alpha_2 x_2 + e = \alpha_1 x_1 + \alpha_2 (x_2 + e_2)$ for some $e$ and $e_2$. Note that, since we are using a field of size $2^n$, multiplying an element from the field by a randomly chosen constant has the effect of randomizing the product. Here, we consider two cases:
\begin{itemize}
\item{\it{Case 1:}} If $x_2 + e_2 \notin \tilde{X}_2$, then $v_3$ fails $v_1$'s check.
\item{\it{Case 2:}} If $x_2 + e_2 \in \tilde{X}_2$, then $v_3$ passes $v_1$'s check; however, $v_3$ is unlikely to pass $v_2$'s check. This is because $\alpha_1 x_1 + \alpha_2 (x_2 + e_2) = \alpha_1 x_1 + \alpha_2 x_2 + \alpha_2 e_2 = \alpha_1 (x_1 + e_1) + \alpha_2 x_2$ for some $e_1$. Here, for uniformly random $\alpha_1$ and $\alpha_2$, $e_1$ is also uniformly random. Therefore, the probability that $v_3$ will pass is the probability that the uniformly random vector $x_1 + e_1$ belongs to $\tilde{X}_1 = \{x\ |\ h(x) = h(x_1)\} \cap B(\mathbf{\tilde{x}_1}, r_{1\rightarrow 2})$ where $v_2$ overhears $\mathbf{\tilde{x}_1}$ from $v_1$, and the probability that the interference channel from $v_1$ to $v_2$ outputs $\mathbf{\tilde{x}_1}$ given $\mathbf{x} \in B(\mathbf{\tilde{x}_1}, r_{1\rightarrow 2})$ is greater than $1-\epsilon$.
    \begin{align*}
    \mathbf{P}(\text{A malicious } v_3 \text{ passes $v_2$'s check}) &= \mathbf{P}(x_1 + e_1 \in \tilde{X}_1)\\
     &= \frac{Vol(\tilde{X}_1)}{2^n},
    \end{align*}
    where $Vol(\cdot)$ is equal to the number of $\{0,1\}$-vectors in the given set. Since $Vol(B(x,r)) = \sum_{k=0}^r \binom{n}{k} \leq 2^n$, and the probability that $h(x)$ is equal to a given value is $\frac{1}{2^h}$, $Vol(\tilde{X}_1)$ is given as follows:
    \[
    Vol(\tilde{X}_1) = \frac{ Vol(B(\tilde{x}_1, r_{1\rightarrow 2}))}{2^h} = \frac{\sum_{k=0}^{r_1\rightarrow 2} \binom{n}{k}}{2^h}.
    \]
\end{itemize}
From $v_1$'s perspective, the probability that a $\mathbf{z} \in \tilde{X}_3$ passes the checks, $\mathbf{P}(\mathbf{z} \text{ passes check})$, is: \vspace*{-0.1cm}\[0\cdot \mathbf{P}(x_2 + e_2 \notin \tilde{X}_2) + \frac{\sum_{k=0}^{r_{1\rightarrow 2}} \binom{n}{k}}{2^{(h+n)}}\cdot \mathbf{P}(x_2 + e_2 \in \tilde{X}_2).\] Similarly, $\mathbf{P}(x_2 + e_2 \in \tilde{X}_2) = \frac{\sum_{k=0}^{r_{2\rightarrow 1}} \binom{n}{k}}{2^{(h+n)}}$, and $ Vol( \tilde{X}_3) = \frac{\sum_{k=0}^{r_{3\rightarrow 1}} \binom{n}{k}}{2^{h}}$. Then, the probability that $v_3$ is undetected from $v_1$'s perspective is the probability that \emph{at least one} $\mathbf{z} \in \tilde{X}_3$ passes the check:
\begin{align*}
\mathbf{P}(&\text{A malicious } v_3 \text{ is undetected from $v_1$'s perspective}) \\
&= \min\{1, \mathbf{P}(\mathbf{z} \text{ passes check})\cdot Vol( \tilde{X}_3)\}
\end{align*}
Note that $\mathbf{P}(\mathbf{z} \text{ passes check})\cdot Vol( \tilde{X}_3)$ is the expected number of $\mathbf{z} \in \tilde{X}_3$ that passes the check; thus, given a high enough $\mathbf{P}(\mathbf{z} \text{ passes check})$, would exceed 1. Therefore, we take $\min\{1, \mathbf{P}(\mathbf{z} \text{ passes check})\cdot Vol( \tilde{X}_3)\}$ to get a valid probability. This proves the statement.
\end{proof}

\begin{lemma}\label{thm:s2}
$\mathbf{P}($A malicious $v_3$ is undetected from $v_2$'s perspective$)$ is:
\[
\min\biggl\{1, \frac{\sum_{k=0}^{r_{1\rightarrow 2}} \binom{n}{k}}{2^{(h+n)}}\cdot \frac{\sum_{k=0}^{r_{2\rightarrow 1}} \binom{n}{k}}{2^{(h+n)}}\cdot \frac{\sum_{k=0}^{r_{3\rightarrow 2}} \binom{n}{k}}{2^{h}}\biggl\},
\]
where $v_2$ overhears $\mathbf{\tilde{x}_3}$ from $v_3$, and the probability that the interference channel from $v_3$ to $v_2$ outputs $\mathbf{\tilde{x}_3}$ given $\mathbf{x} \in B(\mathbf{\tilde{x}_3}, r_{3\rightarrow 2})$ is greater than $1-\epsilon$.

\end{lemma}
\begin{proof}
By similar analysis as in proof of Lemma \ref{thm:s1}.
\end{proof}

\begin{theorem}\label{thm:main}
The probability of misdetection, $\beta$, is:
\[
\beta = \min\biggl\{1, \frac{\sum_{k=0}^{r_{1\rightarrow 2}} \binom{n}{k}}{2^{(h+n)}}\cdot \frac{\sum_{k=0}^{r_{2\rightarrow 1}} \binom{n}{k}}{2^{(h+n)}}\cdot \frac{1}{2^h} \sum_{k=0}^{r} \binom{n}{k}\biggl\},
\]
where $r = \min\{r_{3\rightarrow 1},r_{3\rightarrow 2}\}$.
\end{theorem}
\begin{proof}
The probability of misdetection is the minimum of the probability that $v_1$ and $v_2$ misdetecting malicious $v_3$. Therefore, by Lemma \ref{thm:s1} and \ref{thm:s2}, the statement is true.
\end{proof}

Theorem \ref{thm:main} shows that the probability of misdetection $\beta$ decreases with the hash size, as the hashes restrict the space of consistent codewords. In addition, since $r_{1\rightarrow 2}$, $r_{2\rightarrow 1}$, $r_{3\rightarrow 1}$, and $r_{3\rightarrow 2}$ represent the uncertainty introduced by the interference channels, $\beta$ increases with them. Lastly and the most interestingly, $\beta$ decreases with $n$, since $\sum^r_{k=0} \binom{n}{k} < 2^n$ for $r < n$. This is because network coding randomizes the messages over a field whose size is increasing exponentially with $n$, and this makes it difficult for an adversary to introduce errors without introducing inconsistencies.

Note that we can apply Theorem \ref{thm:main} even when $v_1$ and $v_2$ cannot overhear each other. In this case, both $r_{1\rightarrow 2}$ and $r_{2 \rightarrow 1}$ equal to $n$, giving the probability of misdetection, $\beta = \min\{1, \sum_{k=0}^r \binom{n}{k}/8^h\}$ where $r = \min\{r_{3\rightarrow 1},r_{3\rightarrow 2}\}$. Here, $\beta$ highly depends on $h$, the size of the hash, as $v_1$ and $v_2$ are only using their own message and the overheard hashes.

The algebraic approach results in a nice analysis with exact formulae for $\gamma$ and $\beta$. In addition, these formulae are conditional probabilities; as a result, they hold regardless of a priori knowledge of whether $v_3$ is malicious or not. However, this approach is not very extensible as the number of ``reasonable'' messages grows exponentially with $m$.

\section{Conclusion and Future Work}\label{sec:conclusion}
We proposed a scheme, the \emph{algebraic watchdog} for coded networks, in which nodes can verify their neighbors probabilistically and police them locally by means of overheard messages. We presented a graphical model for two-hop networks to explain how a node checks its neighbors; as well as compute, analyze, and potentially approximate the probabilities of misdetection/false detection. We also provided an algebraic analysis of the performance using an hypothesis testing framework, which gives exact formulae for the probabilities.

Our ultimate goal is to design a network in which the participants check their neighborhood locally to enable a secure global network –- \ie a self-checking network. There are several avenues for future work, of which we shall list only a few. First, there is a need to develop models and frameworks for the algebraic watchdog in general topology as well as multi-hop networks. In addition,
possible future work includes developing inference methods and approximation algorithms for nodes to decide efficiently whether they believe their neighbor is malicious or not.

%
%

\bibliography{References}

\bibliographystyle{IEEEtran}

\end{document}